\documentclass[a4paper,USenglish,cleveref,thm-restate]{lipics-v2021}
\hideLIPIcs
\nolinenumbers

\makeatletter
\renewcommand\@oddfoot{
	\hfil
	\rlap{%
		\vtop{%
			\vskip10mm
			\colorbox[rgb]{0.99,0.78,0.07}
			{\@tempdima\evensidemargin
				\advance\@tempdima1in
				\advance\@tempdima\hoffset
				\hb@xt@\@tempdima{%
					\textcolor{darkgray}{\normalsize\sffamily
						\bfseries\quad
						\expandafter\textsolittle\expandafter{
							arXiv.org}}%
					\strut\hss}}}}
}
\makeatother


\authorrunning{F.~Frei, R.~Gelles, A.~Ghazy, A.~Nolin}
\Copyright{Fabian Frei, Ran Gelles, Ahmed Ghazy, Alexandre Nolin}

\author{Fabian Frei}{CISPA Helmholtz Center for Information Security, Germany}{fabian.frei@cispa.de}{https://orcid.org/0000-0002-1368-3205}{}
\author{Ran Gelles}{Bar-Ilan University, Israel}{ran.gelles@biu.ac.il}{https://orcid.org/0000-0003-3615-3239}{Partially supported by a grant from the United States-Israel Binational Science Foundation (BSF), Jerusalem, Israel, Grant No.\@ 2020277. He would also like to thank the CISPA Helmholtz Center for Information Security for hosting him while part of this research was done.}
\author{Ahmed Ghazy}{CISPA Helmholtz Center for Information Security, Germany}{ahmed.ghazy@cispa.de}{https://orcid.org/0009-0009-7414-5871}{}
\author{Alexandre Nolin}{CISPA Helmholtz Center for Information Security, Germany}{alexandre.nolin@cispa.de}{https://orcid.org/0000-0002-3952-0586}{}

\usepackage{anyfontsize}

\usepackage[utf8]{inputenc}
\usepackage{microtype}
\usepackage{mathtools,mathrsfs,amsthm}
\usepackage{thmtools}
\usepackage{xspace}
\usepackage{tikz}
\usepackage{algorithm}
\usepackage{algorithmicx}
\usepackage[noend]{algpseudocode}

\crefname{prop}{Property}{Properties}

\algnewcommand{\IIf}[1]{\State\algorithmicif\ #1\ \algorithmicthen}
\algnewcommand{\EndIIf}{\unskip\ \algorithmicend\ \algorithmicif}
\algnewcommand{\ElseIIf}[1]{\algorithmicelse\ #1}

\algnewcommand{\IWhile}[1]{\State\algorithmicwhile\ #1\ \algorithmicdo}
\algnewcommand{\EndIWhile}{\unskip\ \algorithmicend\ \algorithmicwhile}

\newcommand{\directionformat}[1]{\ensuremath{\mathsf{#1}}\xspace}
\global\long\def\CW{\directionformat{CW}}
\global\long\def\CCW{\directionformat{CCW}}

\newcommand{\N}{\mathbb{N}}

\newcommand{\ID}[1]{\ensuremath{\mathsf{ID}_{#1}}}
\newcommand{\IDsup}[2]{\ensuremath{\ID{#1}^{\mathrm{#2}}}}
\newcommand{\IDcw}[1]{\IDsup{#1}{cw}}
\newcommand{\IDccw}[1]{\IDsup{#1}{ccw}}

\newcommand{\maxID}{\ensuremath{\ID{\max}}\xspace}

\newcommand{\vmax}{\ensuremath{\ell}\xspace}
\newcommand{\Vmax}{\ensuremath{V_{\max}}\xspace}

\newcommand{\stringlength}{\ensuremath{s}\xspace}

\DeclarePairedDelimiter{\set}{\lbrace}{\rbrace}

\DeclarePairedDelimiter{\card}{\lvert}{\rvert}

\newcommand{\SRDir}[3]{{{{#1}{#2}}\ensuremath{(#3)}}}

\newif\ifColors
\Colorsfalse
\ifColors
\newcommand{\SendCW}[1]{\textcolor{red}{\SRDir{send}{\CW}{#1}}}
\newcommand{\SendCCW}[1]{\textcolor{blue}{\SRDir{send}{\CCW}{#1}}}
\newcommand{\RecvCW}[1]{\textcolor{red}{\SRDir{recv}{\CW}{#1}}}
\newcommand{\RecvCCW}[1]{\textcolor{blue}{\SRDir{recv}{\CCW}{#1}}}
\else
\newcommand{\SendCW}[1]{\textcolor{black}{\SRDir{send}{\CW}{#1}}}
\newcommand{\SendCCW}[1]{\textcolor{black}{\SRDir{send}{\CCW}{#1}}}
\newcommand{\RecvCW}[1]{\textcolor{black}{\SRDir{recv}{\CW}{#1}}}
\newcommand{\RecvCCW}[1]{\textcolor{black}{\SRDir{recv}{\CCW}{#1}}}
\fi

\newcommand{\SendPort}[2]{\SRDir{send}{\ensuremath{\mathrm{Port}_{#1}}}{#2}}
\newcommand{\RecvPort}[2]{\SRDir{recv}{\ensuremath{\mathrm{Port}_{#1}}}{#2}}
\newcommand{\Port}[1]{\ensuremath{\mathrm{Port}_{#1}}}

\newcommand{\tot}[2]{\ensuremath{#1_{\mathsf{#2}}}}

\def\totrecvcw{\tot{\rho}{cw}}
\def\totrecvccw{\tot{\rho}{ccw}}
\def\totsentcw{\tot{\sigma}{cw}}
\def\totsentccw{\tot{\sigma}{ccw}}

\newcommand{\totrecvport}[1]{{\rho}_{#1}}

\def\state{\textit{state}}

\newcommand{\leader}{\ensuremath{\mathsf{Leader}}\xspace}
\newcommand{\nonleader}{\ensuremath{\mathsf{Non\textsf{-}Leader}}\xspace}

\title{Content-Oblivious Leader Election on Rings}

\begin{document}

\maketitle

\begin{abstract}
In \emph{content-oblivious computation}, 
$n$ nodes wish to compute a given task over an asynchronous network that suffers from an extremely harsh type of noise, which corrupts the content of all messages across all channels. 
In a recent work, Censor-Hillel, Cohen, Gelles, and Sela (Distributed Computing, 2023) showed how to perform arbitrary computations in a content-oblivious way in 2-edge connected networks but only if the network has a distinguished node (called \emph{root}) to initiate the computation. 

Our goal is to remove this assumption, which was conjectured to be necessary. 
Achieving this goal essentially reduces to performing a content-oblivious leader election  
since an elected leader can then serve as the root required to perform arbitrary content-oblivious computations. 
We focus on ring networks, which are the simplest 2-edge connected graphs. On \emph{oriented} rings, we obtain a leader election algorithm with message complexity $O(n \cdot \maxID)$, where $\maxID$ is the maximal assigned ID.
As it turns out, this dependency on $\maxID$ is inherent: 
we show a lower bound of $\Omega(n \log(\maxID/n))$ messages for content-oblivious leader election algorithms. 
We also extend our results to \emph{non-oriented} rings, where nodes cannot tell which channel leads to which neighbor. 
In this case, however, the algorithm does not terminate but only reaches quiescence.
\end{abstract}

\section{Introduction}
\label{sec:intro}

The field of distributed computing is rich with models helping us understand different types of architectures and computational hardness by making different kinds of assumptions about how computations are carried out. A recent work by 
Censor-Hillel, Cohen, Gelles, and Sela~\cite{ccgs23} introduced a particularly weak computational model coined 
\emph{fully defective networks}.
This model considers 
an asynchronous network in which messages may be fully corrupted and thus not carry any information beyond their sheer existence. 
Algorithms designed for this model cannot rely 
in any way on possible contents of messages, and thus are named \emph{content-oblivious}. 
Instead of relying on the content of messages or on their time of arrival (since arbitrary delays may occur in asynchronous networks), such algorithms depend solely on the \emph{order} in which messages arrive from different neighbors.

The aforementioned work by Censor-Hillel et al.~\cite{ccgs23} showed that
any computation possible in the asynchronous setting with reliable content-carrying messages can be simulated in the fully defective setting 
under two assumptions: that the network is 2-edge connected and that there is a distinguished leader (which they call the \emph{root node}) in the network.
While the same paper showed 2-edge connectivity to be essential for any kind of nontrivial computation in fully defective networks, 
the question of the necessity of a pre-existing leader was not settled. Censor-Hillel et al.\@ conjectured that general computations in fully defective networks do indeed require the existence of such a pre-elected leader.

In this paper, we \emph{disprove} this conjecture, at least for the most fundamental type of 2-edge connected topologies, namely, rings. 
We design a content-oblivious algorithm that successfully performs a leader election in oriented rings.
\begin{restatable}{theorem}{LEoriented}
\label{thm:LE-oriented}
    There is a \emph{quiescently terminating} content-oblivious algorithm of message complexity $n(2\cdot \maxID + 1)$
    that elects a leader in oriented rings of $n$ nodes with unique IDs.
\end{restatable}
Here, $\maxID$ denotes the maximal ID assigned to a node in the network, and \emph{quiescent termination} refers to the valuable property of our algorithm that nodes do not receive messages after termination. 
We discuss quiescent computations and the issue of composability in more detail  in Section~\ref{sec:termination}.

We further extend our result to the case where the ring is non-oriented, albeit under a weaker definition of computation. 
Instead of termination, we only require \emph{stabilization}, which means that every node eventually settles on a decision (to be or not to be a leader, in our case) that is never revised again. However, the nodes might not know whether they have achieved their stable output already and remain ready to receive and potentially send messages forever. This potential is not actualized for a \emph{quiescently} stabilizing algorithm, where all messaging activity ceases after finite time. Our algorithm not only elects a leader in such a manner but also orients the ring.

\begin{restatable}{theorem}{orientRing}
\label{thm:orient-ring}
     There is a quiescently stabilizing content-oblivious algorithm of message complexity $n(2\cdot\maxID + 1)$
     that elects a leader and orients a non-oriented ring of $n$ nodes with unique~IDs.
\end{restatable}

Both of our leader election algorithms, for oriented and non-oriented rings, have message complexity~$O(n \cdot \maxID)$. 
The term $\maxID$ is not very common in message complexities of algorithms.\footnote{Some algorithms depend on the IDs assigned to nodes. However, it is common for the complexity analysis to assume that all IDs are of length $O(\log n)$ bits, making the dependence on the  IDs implicit.}  
Somewhat surprisingly, our analysis shows that this term is inherent to content-oblivious algorithms. We prove the following lower bound.
\begin{theorem}
\label{thm:LB}
    Any deterministic terminating content-oblivious algorithm for leader election in rings with unique IDs sends at least $n \lfloor\log(\maxID/n)\rfloor$ messages.
\end{theorem}
Note that this $\Omega(n \log(\maxID/n))$ lower bound implies that the number of messages in a ring of size $n$ is unbounded---
we can always increase the message complexity by assigning larger IDs, even when $n$ is a small constant.

\subsection{Quiescence and Composability}
\label{sec:termination}
As mentioned above,
we concatenate our content-oblivious leader election algorithm (Theorem~\ref{thm:LE-oriented}) with the root-dependent universal content-oblivious algorithm \cite[Thm.~1]{ccgs23} to obtain the following powerful corollary.
\begin{restatable}{corollary}{GeneralComputations}
\label{thm:GeneralComputations}
    Assuming unique IDs, any asynchronous algorithm on rings can be simulated in a fully defective oriented ring.
\end{restatable}
However, some subtleties arise when concatenating algorithms in the content-oblivious setting. 
With reliable message content, each message could be tagged, indicating the algorithm it belongs to. 
However, when concatenating content-oblivious algorithms, messages sent by the first algorithm may be mistaken for ones sent by the second algorithm, and vice versa.

To make this composition work, we require two properties: (1) \emph{termination}, i.e., nodes should have a distinct point in time where they end the first algorithm and switch to the second one, and (2) \emph{message-algorithm attribution}, i.e., 
while a node executes an algorithm, it only ever receives messages generated by nodes while they were executing the same algorithm. 

Our algorithm for oriented rings (Theorem~\ref{thm:LE-oriented}) achieves correct message-algorithm attribution
(when composed with the scheme of~\cite{ccgs23}) 
by combining two mechanisms.
First, it features 
a \emph{quiescent} termination.
In particular, any node terminates eventually, and at that time no messages are in delivery towards that node anymore, nor will any message be sent to that node after its termination. 
Secondly, the nodes terminate in order, so that the leader is the last to terminate. This makes our algorithm easy to compose with the algorithm of~\cite{ccgs23}, by replacing the act of termination with the act of switching to the second algorithm. The leader, which is the last to terminate the first algorithm, is the node that initiates the computation of the scheme in~\cite{ccgs23} (i.e., it acts as the root), and at the time it sends its first message, we are guaranteed that all other nodes have already switched to that algorithm. 

As a matter of fact,
quiescent termination could be relaxed when composing general algorithms:
If we have a bound~$r$ on the number of messages of the first algorithm that might reach a node after it transitions to the second algorithm, we could still concatenate any algorithm in an \emph{altered form} where nodes send $r+1$ copies of each message, and  process arriving messages in groups of $r+1$ messages as well. However, this clearly leads to an undesired $r$-fold increase in the message complexity of the composed algorithm.

\medskip

On the other hand, our algorithm for non-oriented rings (Theorem~\ref{thm:orient-ring}) does not terminate and cannot be composed with other algorithms. As mentioned above, we only require the algorithm to reach quiescence in this case.

\subsection{Related Work}
\label{sec:relwork}
Leader election is a fundamental task that has been studied by the distributed computing community since the 1970s.
Simple leader election algorithms for asynchronous rings were proposed by Le Lann~\cite{LeLann77} and by Chang and Roberts~\cite{CR79}. These algorithms employ $O(n^2)$ messages to ensure that all nodes yield to the node with the maximal ID, which becomes the leader. Later work~\cite{HS80,DKR82,Peterson82} improved this down to~$O(n \log n)$ messages. This complexity is tight in asynchronous rings~\cite{Burnes80,FL87}. In \emph{synchronous} rings, leader election can be performed by communicating only $O(n)$ messages~\cite{FL87,EKCK91}. 

For the case of anonymous rings, i.e., with identical nodes without IDs,
Angluin~\cite{Angluin80} proved that symmetry cannot be broken,
and thus no leader election algorithm exists. 
Attiya, Snir, and Warmuth~\cite{ASW88} examined anonymous asynchronous rings further and characterized computable tasks. In particular, they showed that many tasks, including ring orientation, require communicating $\Omega(n^2)$~messages. Attiya and Snir~\cite{AS91} gave tight bounds on the complexities of randomized algorithms.
Relaxed notions of ring orientation were given by Attiya, Snir, and Warmuth~\cite{ASW88}, where the orientation is either consistent with all nodes or alternating between any two neighbors,
and 
by Syrotiuk and Pachel~\cite{SP88}, where nodes determine whether a majority agrees on the same orientation or not. 
Orienting a ring and leader election when $n$ is known to the nodes is presented by Flocchini et al.~\cite{FKKLS04}.

The question of termination in anonymous rings was explored by
Itai and Rodeh~\cite{IR90},
who discovered that a network cannot compute its size~$n$ by a terminating algorithm. As a consequence, a leader election algorithm that terminates is impossible, too, since it is trivial to learn $n$ once a leader is given.
On the other hand, if the nodes know $n$ or even an upper bound on~$n$, a randomized leader election algorithm that terminates exists~\cite{IR90}. Afek and Matias~\cite{AM94} give various leader election algorithms that trade off knowledge, termination, and number of sent messages.

\smallskip
Censor-Hillel, Gelles, and Haeupler~\cite{CGH19} designed a content-oblivious BFS algorithm as a pre-processing step for their distributed interactive coding scheme~\cite{gelles17}. 
Building on that idea, Censor-Hillel et al.~\cite{ccgs23} introduced the concept of content-oblivious computation and proved that general computations are only possible over 2-edge connected networks.
They designed a compiler that converts any asynchronous algorithm into a content-oblivious one, assuming a pre-existing
leader. They also gave explicit algorithms for content-oblivious DFS, ear decomposition, and (generalized) Hamiltonian cycle construction.

Computation with fully corrupted messages has been comparatively more studied in the synchronous setting, where the presence or absence of a message in a given round can still be used to send one bit of information~\cite{SW90}. 
A notable example is the \emph{Beeping} model, introduced by Cornejo and Kuhn~\cite{CK_disc10}, in which each message sent by a node is received by all its neighbors, and nodes can only distinguish between receiving no message and receiving at least one message.
Among other problems, leader election has been studied in this setting in a sequence of works~\cite{GH_soda13,FSW_disc14,DBB_disc18,CD_tcs19}.
Similar to the fully defective setting, compilers were devised for transforming algorithms with stronger communication primitives into algorithms for the Beeping model~\cite{DBB_podc20,AGL_iandc22,Davies_podc23a}.

\subsection{Organization}
Section~\ref{sec:preliminaries} formally introduces the content-oblivious setting and some other notions and notations.
In Section~\ref{sec:LE-oriented} we give our leader election algorithm in oriented rings, deferring some proofs to \cref{sec:deferred-analysis-terminating-le}. Section~\ref{sec:LE-non-oriented} discusses non-oriented rings. 
Our lower bound on the number of messages sent by a content-oblivious leader election algorithm is presented in Section~\ref{sec:lower-bounds}.
In Appendix~\ref{sec:anonymousRings}  we discuss anonymous rings and design a randomized quiescently stabilizing algorithm with a high probability of success for both electing a leader and orienting a non-oriented ring.
Section~\ref{sec:conclusion} concludes our work with a brief summary and suggests a few follow-up questions.
\section{Preliminaries}
\label{sec:preliminaries}

\subparagraph*{Notations.}
For an integer $n$, we denote the set $\{1,2,\ldots,n\}$ by $[n]$.
All logarithms are binary. For a variable $\mathit{var}$ and a node $v$, we let $\mathit{var}[v]$ denote the value the variable $\mathit{var}$ holds at the node~$v$.
By the term \emph{with high probability} we mean a probability of at least $1-n^{-O(1)}$.

\subparagraph*{The content-oblivious computation model.}
Consider a distributed network $G=(V,E)$ with $n=|V|$ nodes. We usually assume that each node $v\in V$ is assigned a unique ID  (unless otherwise stated), usually denoted by~$\ID{v}\in\N$. 
We denote the largest ID assigned to a node in the network 
by~$\maxID$, that is,
$\maxID := \max_{v\in V}{\ID{v}}$, and by $\vmax$ the node possessing this ID.
Apart from the IDs, the nodes are identical. 
Algorithms are \emph{uniform} by default, i.e., the nodes neither know the size of the network~$n$ nor any bound on it; for non-uniform algorithms, the nodes may be equipped with such knowledge. 

In this paper, we consider rings, i.e., connected graphs where all nodes have degree 2.
Neighboring nodes communicate by sending messages to each other. We assume that all messages are subject to corruption: the content of any message is completely erased by noise, resulting in an empty message of length 0, which we call a \emph{pulse}. 
Computations that ignore any potential content of a message and thus work only with pulses are called \emph{content-oblivious}. 

Further, the network is \emph{asynchronous}: the time it takes a pulse to travel through a channel and arrive at its end is unpredictable; the delays are unbounded but always finite. Pulses cannot be dropped or injected by the channel. In such an asynchronous network, the nodes possess neither a common clock nor any notion of time, and they are assumed to be \emph{event-driven}. This means that a node may act once right in the beginning of the computation and from then on only upon receiving a pulse.  
As a function of 
its own ID, the previously received pulses, and possibly its own source of randomness, it can then change its state and, for each connected channel, send any number of pulses. 
Most of our algorithms are deterministic, except for the ones in \cref{sec:anonymousRings}, where we consider a special case where nodes do not have IDs and use randomness in order to generate them.
We say that an algorithm \emph{terminates} (sometimes also referred to as process termination, explicit termination, or termination detection in the literature) if for each node there is a time at which it has decided on an output  
and entered a terminating state. Once in a terminating state, a node ignores all incoming pulses and does not send any new ones. 
We say that an algorithm has \emph{quiescent termination} if at the time the last node terminates it is guaranteed that no pulses are still in transit.
An algorithm's \emph{message complexity} is the total number of messages (pulses) it sends during a worst-case computation until all nodes have terminated or until the network has reached quiescence.

\subparagraph*{Ring's orientation.}
In a ring, each node communicates with its two neighbors via \Port0 and \Port1. 
Consider a pulse re-sent from \Port1
by every node receiving it. 
If such a pulse passes through all edges
(i.e., is never reflected by a node with misaligned ports), 
we call the ring \emph{oriented} and the pulse \emph{clockwise} (\CW). \Port1 of each node is then called its \CW port, 
leading to its \CW neighbor.  Counterclockwise (\CCW) is defined analogously via \Port0. Note that \CW pulses are sent from \CW ports but arrive at \CCW ports, and vice versa.

However, rings may not be oriented to begin with. In a \emph{non-oriented} ring, there is no guarantee that  \Port0 and \Port1 of a node are aligned with a clockwise or counterclockwise walk on the ring.
See Fig.~\ref{fig:ring-orientation} for a demonstration.
Algorithms for non-oriented rings must work correctly for all assignments of the nodes' ports. In this case the \CW/\CCW direction is local per node, and we will use the notion of \Port0 and \Port1 to avoid confusion. 

\begin{figure}[ht]
    \centering
    \includegraphics[width=0.3\textwidth]{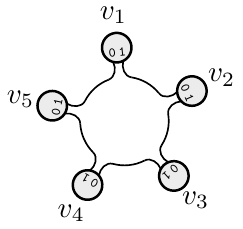}
    \rule{0.2\textwidth}{0pt}
    \includegraphics[width=0.3\textwidth]{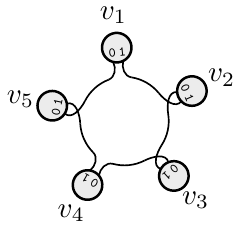}
    \caption{An oriented ring (left) and a non-oriented one (right).}
    \label{fig:ring-orientation}
\end{figure}

\section{Leader Election in an Oriented Ring}
\label{sec:LE-oriented}

In this section, we consider the problem
of electing a leader in an oriented ring of size unknown to the nodes. 
Recall that in the  \emph{leader election task},
nodes are required to terminate with an output: \leader or \nonleader, where  a single node $\vmax \in V$ outputs
\leader, while all other nodes 
must output \nonleader.
We design an algorithm to elect a leader that features quiescent termination and prove the following.
\LEoriented*

For the algorithms in this section, we will use the following methods for sending and receiving pulses. 
The method
\SendCW{} sends one pulse over the \CW channel.
The method \RecvCW{} checks whether pulses are waiting in the \CW incoming queue. If no pulse is in the queue, the method returns~0. Otherwise, it consumes a single pulse from the queue and returns~1. The methods \SendCCW{} and \RecvCCW{}  are the analogous \CCW versions.
Moreover, for each node $v$, we introduce counters $\totrecvcw$ and $\totsentcw$ for the total
number of received and sent \CW pulses, respectively.
Likewise, we introduce $\totrecvccw$ and $\totsentccw$ for \CCW pulses.
We assume each of the four methods above
updates those counters with every received or sent pulse. More precisely,
every node $v$ has variables $\totrecvcw[v]$ and $\totsentcw[v]$ initially set to $0$. 
Every time \RecvCW{} processes a pulse from the queue of incoming pulses or \SendCW{} sends a pulse, the counters
increase
according to 
$\totrecvcw[v] \gets \totrecvcw[v] + 1$ and 
$\totsentcw[v] \gets \totsentcw[v] + 1$, respectively.

\subsection{Warm-up: Leader Election Without Termination}
\label{sec:warmup}
To demonstrate some of the main ideas of our algorithm, let us begin with a simple quiescently stabilizing algorithm that uses only clockwise pulses and elects a leader in an oriented ring. We emphasize that this algorithm is
non-terminating.
In this algorithm (see Algorithm~\ref{alg:nonterminating_le}), each node starts by sending one pulse and then
relays every received pulse in the same direction, except for the single time when the number of received pulses reaches its own ID. In this event, the node does not relay this one pulse and assigns itself the state of \leader, at least temporarily. However, any pulses received after this are relayed again, and revert the node to being a \nonleader.

The main intuition behind this algorithm is that each node will eventually have sent and
received exactly $\maxID{}$ pulses:
$n$ pulses are being generated at the initialization, and these pulses keep circulating in the ring, increasing the counter of received pulses until some node has received as many pulses as its ID. At this point, the node removes one pulse, and we are left with $n-1$ pulses in circulation. 
Except for the node with \maxID, every node will receive more pulses than its ID,
ensuring it eventually becomes a follower. 
This continues until all nodes, except the one with~\maxID, have removed a single pulse from the circulation and declared themselves followers. 
The last remaining pulse keeps circulating until all nodes have received exactly $\maxID$ pulses. As soon as the last node receives its $\maxID$-th pulse, it sets itself as a leader and removes the last remaining pulse. Since the network no longer contains any pulse, which we call quiescence, the states of the event-driven nodes remain unchanged as well. Note, however, that nodes do not terminate since they do not know whether the ring has achieved this quiescent state or not, i.e., whether some pulses are still in transit. 

\begin{algorithm}[ht]\caption{Quiescently Stabilizing Leader Election for Node $v$}
\label{alg:nonterminating_le}
    \begin{algorithmic}[1]
        \State{\SendCW{}}
        \While {\textsf{true}}
            \If {\RecvCW{} returns 1}
                \If{$\totrecvcw = \ID{v}$} \label{line:no_send_cond}
                 \State   {$\state \gets \leader$} 
                \Else \Comment{$v$ acts as a relay unless $\totrecvcw = \ID{v}$}
                    \State{$\state \gets \nonleader$}
                    \State{\SendCW{}}
                \EndIf
            \EndIf
        \EndWhile
    \end{algorithmic}
\end{algorithm}

Even though \cref{alg:nonterminating_le} does not terminate, analyzing it will help us design our terminating leader election algorithm (Algorithm~\ref{alg:terminating_le}) later on.
We begin by proving some key invariants of Algorithm~\ref{alg:nonterminating_le}, which are also used to show its correctness.

\begin{lemma}\label{lem:le_invariants}
    For every node $v$ running \Cref{alg:nonterminating_le},
    the following invariants hold at the end of each iteration of the main loop:
    \begin{enumerate}
        \item If $\totrecvcw < \ID{v}$, then $\totsentcw  = \totrecvcw + 1$, i.e., $v$ has sent exactly one pulse
        more than it has received. \label[lemma]{lem:le_invariants:prop:one}
        \item If $\totrecvcw \geq \ID{v}$, then $\totsentcw  = \totrecvcw$, i.e., $v$ has sent exactly as many pulses as it has received. \label[lemma]{lem:le_invariants:prop:two}
    \end{enumerate}
\end{lemma}
\begin{proof}
    Fix an arbitrary node $v\in V$. For a variable $var$ and $i\in\N$, we write $var(i)$ to denote the value of~$var$ for this node~$v$ at the end of the $i$-th iteration. We also
    use $var(0)$ to denote the initial value of $var$ before the first iteration.
    We prove the statement by induction over $i$.
    Before the first iteration, exactly one initial pulse is sent, and $\totsentcw(0) = 1$.
    Due to 
    $\totrecvcw(0) = 0 < \ID{v}$, (2) is vacuously true, and (1) holds too since $\totsentcw(0) = 1 = \totrecvcw(0) + 1$, proving the induction basis.
    
    Assume now that the invariants hold true at the beginning of some iteration $i \geq 1$ or, equivalently, at the end of iteration $i-1$. We proceed to show that they hold at the end of iteration $i$ again.
    If no pulse is received during iteration $i$,
    then the variables in question remain unchanged, and the invariants still hold after the iteration.
    If a pulse is received, however, there are three cases:
    
    \begin{enumerate}
        \item If $\totrecvcw({i-1}) < \ID{v} - 1$, then $\totrecvcw(i) = \totrecvcw(i-1) + 1 < \ID{v}$, and we need
        to prove the first invariant for iteration $i$. Since the condition on \Cref{line:no_send_cond} is not met, a pulse is sent, and we get $\totsentcw(i) = \totsentcw(i-1) + 1 = \totrecvcw(i-1) + 2 = \totrecvcw(i) + 1$, where the second equality follows from the first invariant for iteration $i-1$. Thus, the invariants still hold.
    
        \item If $\totrecvcw({i-1}) = \ID{v}-1$, then $\totrecvcw({i}) = \ID{v}$.
        Thus, in iteration $i$, by the condition on \Cref{line:no_send_cond}, no pulses are sent, so
        $\totsentcw({i}) = \totsentcw({i-1}) = \totrecvcw(i-1) + 1 = \totrecvcw(i)$, where the second
        equality is due to the first invariant for iteration $i-1$.
        Thus, we have $\totsentcw({i}) = \totrecvcw({i})$ and $\totrecvcw({i}) \geq \ID{v}$, and the invariants hold.

        \item If $\totrecvcw({i-1}) > \ID{v} - 1$, we have $\totrecvcw(i) > \ID{v}$.
        The condition on \Cref{line:no_send_cond} is not met in iteration $i$, and a pulse is relayed.
        We thus have $\totsentcw(i) = \totsentcw(i-1) + 1 = \totrecvcw(i-1) + 1 = \totrecvcw(i)$, with the
        second equality again due to the second invariant for iteration $i-1$.
        \qedhere
    \end{enumerate}
\end{proof}

\begin{lemma}\label{lem:leader_last}
    Let $\vmax$ be the node with $\ID{\vmax} = \maxID$. If
    $\totrecvcw[\vmax] \geq \ID{\vmax}$ at some point, then $\totrecvcw[v] \geq \ID{v}$
    holds for every node $v$ at this point. That is, $\vmax$ is the last node to 
    satisfy $\totrecvcw[\vmax] \geq \ID{\vmax}$.
\end{lemma}
\begin{proof}
    Consider the iteration where $\totrecvcw[\vmax] = \ID{\vmax}$ holds.
    Assume towards contradiction that the statement does not hold, and let $x_k\in V$ be the closest predecessor\footnote{i.e., along the \CW path of the network.} of $\vmax$ for which
    $\totrecvcw[x_k] < \ID{x_k}$.
    By \Cref{lem:le_invariants:prop:one},
    $\totsentcw[x_k] = \totrecvcw[x_k] + 1$. 
    Denoting the intermediate nodes by $x_1, \ldots, x_{k-1}$, by \Cref{lem:le_invariants:prop:two}, we
    have $\totsentcw[x_i] = \totrecvcw[x_i]$ for all $i \in [k-1]$.

    As the number of pulses sent by a node upper bounds the number of pulses its neighbor has received,
    we have the following inequality.
    \[
    \totrecvcw[\vmax] \leq \totsentcw[x_1] = \totrecvcw[x_1] \leq \dots \leq \totsentcw[x_k] = \totrecvcw[x_k] + 1.
    \]
    Since $\totrecvcw[\vmax] = \ID{\vmax}$ and  $\totrecvcw[x_k] < \ID{x_k}$, the above inequality implies
    $\ID{\vmax} \leq \ID{x_k}$, which is a contradiction due to the uniqueness of the largest ID and $\ID{\vmax} = \maxID$.
    Therefore, $\vmax$ is the last node meeting the condition with $\totrecvcw[\vmax] = \ID{\vmax}$.
\end{proof}

We now show that quiescence has been reached when $\totrecvcw[v] \geq \ID{v}$ holds for all nodes.

\begin{lemma}\label{lem:sat_to_quiescence}
    If $\totrecvcw[v] \geq \ID{v}$ holds at every node $v$, then the network
    is in quiescence.
\end{lemma}
\begin{proof}
    By \Cref{lem:le_invariants:prop:two}, we get that
    $\totsentcw[v] = \totrecvcw[v]$ holds for all $v$; hence, the total number of sent pulses is equal to
    the total number of received pulses. 
    In particular,
    no pulses
    are in transit (sent but not received).
\end{proof}

In fact, the converse is also true; that is, a necessary condition for quiescence
is that $\totrecvcw[v] \geq \ID{v}$ holds at every node $v$ and becomes a relay.

\begin{lemma}\label{lem:quiescence_to_sat}
    If the network is in quiescence, then
    $\totrecvcw[v] \geq \ID{v}$ holds at every node $v$.
\end{lemma}
\begin{proof}
    If there is quiescence, then $\sum_v{\totsentcw[v]} = \sum_v{\totrecvcw[v]}$.
    
    By \Cref{lem:le_invariants}, every node $v$ has $\totsentcw[v] \geq \totrecvcw[v]$.
    Assuming there is some bad node $b$ with $\totrecvcw[b] < \ID{b}$, then by \Cref{lem:le_invariants:prop:one},
    $\totsentcw[b] = \totrecvcw[b] + 1$. Therefore,
    $\sum_v{\totsentcw[v]} =  \sum_{v \neq b} {\totsentcw[v]} + \totsentcw[b] \geq \sum_v{\totrecvcw[v]} + 1$, so this cannot occur.
\end{proof}

\begin{corollary}\label{cor:quiescence_convergence}
    There is quiescence at some point of time
    if and only if each node $v$ has $\totrecvcw[v] \geq \ID{v}$ at that point of time.
\end{corollary}
\begin{proof}
    The corollary follows directly from \Cref{lem:sat_to_quiescence,lem:quiescence_to_sat}.
\end{proof}

Another equivalent statement to the ones in \Cref{cor:quiescence_convergence} is
that every node has sent and received exactly \maxID pulses.

\begin{lemma}\label{lem:le_equivalences}
    In any execution of \Cref{alg:nonterminating_le}, at any point of time, the following
    statements are equivalent:
    \begin{enumerate}
        \item The network is in quiescence,
        \item $\forall{v}: \totrecvcw[v] \geq \ID{v}$, and
        \item $\forall{v}: \totrecvcw[v] = \totsentcw[v] = \maxID$.
    \end{enumerate}
\end{lemma}
\begin{proof}
    The first two statements are equivalent by \Cref{cor:quiescence_convergence}.
    
    Now, assuming $\totrecvcw[v] = \maxID$ holds for every $v$, then
    so does $\totrecvcw[v] \geq \ID{v}$. Conversely, assume
    $\forall{v}: \totrecvcw[v] \geq \ID{v}$.
    By \Cref{lem:leader_last}, the node $\vmax$ with the largest ID is the last
    to satisfy that inequality. After the iteration where this occurs
    for the first time, we have
    \[
    \totrecvcw[\vmax] = \ID{\vmax} = \maxID.
    \]
    By the first equivalence, there is quiescence, and
    no more pulses are sent or received. Any pulses that have been sent by a node $u$ have
    been received by its neighbor $v$, that is, over all \CW edges~$(u, v)$, it holds that
    \(
    \totsentcw[u] = \totrecvcw[v]\text{.}
    \)
    
    Also, by \Cref{lem:le_invariants:prop:two}, every node $v$ has
    \(
    \totsentcw[v] = \totrecvcw[v]\text{.}
    \)
    As the network forms a ring, combining those equations yields that
    $\totrecvcw[v] = \totsentcw[v] = \maxID$ holds for every $v$.
\end{proof}

Given the above properties, we are ready to prove that certain interesting events occur in every execution
of \Cref{alg:nonterminating_le}.

First, we show that the inequality $\totrecvcw[v] \geq \ID{v}$ is eventually satisfied
by every $v$. Due to the equivalences given by \Cref{lem:le_equivalences},
this leads to a setting where quiescence is achieved, and every node has sent
and received exactly \maxID pulses.

\begin{lemma}\label{lem:le_convergence}
In any execution of \Cref{alg:nonterminating_le},
for every node $v$, there is some iteration,
where $\totrecvcw[v] \geq \ID{v}$ holds.
\end{lemma}
\begin{proof} 
    Let us track the evolution through time of $B$, the set of nodes that have
    not met the condition yet; that is, at every point of time,
    for all $b \in B$, $\totrecvcw[b] < \ID{v}$.
    
    Initially, $B$ contains all nodes, each of which 
    are removed upon meeting the condition. Since a removed node never enters $B$
    again, $|B|$ is monotonically decreasing.

    Consider the point of time where $|B|$ reaches a minimum, so $B$ remains fixed.
    If $|B| = 0$, then there is nothing to prove, so assume $|B| > 0$.
    At that time, maintain values $\Delta_b :=  \ID{b} - \totrecvcw[b] > 0$ for every
    $b \in B$.

    By \Cref{lem:le_invariants:prop:one}, for all $b \in B$, we have that $\totsentcw[b] = \totrecvcw[b] + 1$ always holds.
    Also, from that point on, for every
    $v \notin B$, we have $\totsentcw[v] = \totrecvcw[v]$ by \Cref{lem:le_invariants:prop:two}.

    The number of pulses in transit\footnote{Including pulses that are in some node's queue, but were not processed yet.} at any given time is the difference between the total number of sent and received
    pulses across all nodes. 
    Therefore, since
    $\sum_v{\totsentcw[v]} = \sum_v{\totrecvcw[v]} + |B|$, there are still $|B|$ pulses
    in the network. Since the nodes outside $B$ act as relays, they maintain the number of pulses in transit and, in particular, never remove a pulse from the network. 
    Eventually, some node $b \in B$ must receive a pulse,
    which decreases the difference $\Delta_b$ by one. 
    If $\Delta_b = 0$, then $b$ is removed from $B$, and $|B|$ decreases beyond its
    minimum, a contradiction.
    Otherwise, $b$ forwards a pulse and the number of pulses in transit
    remains~$|B|$, so we
    can re-apply this argument.
    At some point, $\Delta_b$ reaches~$0$
    for some $b \in B$. Thus, 
    $\totrecvcw[b] \geq \ID{b}$ holds for~$b$, and $b$ is removed from~$B$,
    again, a contradiction.
\end{proof}

A result of \Cref{lem:le_convergence}, is that, eventually, all nodes
send and receive exactly $\maxID$ pulses, and no further activity occurs on
the network.

\begin{corollary}\label{cor:sent_recv_maxid}
    In any execution of \Cref{alg:nonterminating_le}, at some point,
    every node has sent and received exactly $\maxID$ pulses, and the network reached quiescence.
\end{corollary}
\begin{proof}
    By \Cref{lem:le_convergence}, there is an iteration, after which
    $\totrecvcw[v] \geq \ID{v}$ holds for all nodes $v$. The statement
    then follows directly from \Cref{lem:le_equivalences}.
\end{proof}

We also have the following trivial upper bound as a direct corollary.

\begin{corollary}\label{cor:recv_bound}
    In any execution of \Cref{alg:nonterminating_le}, for every node $v$, at every point of time
    $\totrecvcw[v] \leq \maxID$.
\end{corollary}
\begin{proof}
    Immediate from \Cref{cor:sent_recv_maxid} and the monotonicity of
    $\totrecvcw[v]$, which is initially~$0$.
\end{proof}

\subsection{Leader Election With Quiescent Termination}
\label{sec:le-quiescent-termination}

Although \Cref{alg:nonterminating_le} is a quiescently stabilizing algorithm for leader election, more ideas are needed in order to achieve this with quiescent termination.
The main idea is to utilize the \CCW channel to notify all nodes when the leader is elected.
The immediate approach would be to leverage the event $\totrecvcw[v] = \ID{v}$, which signifies the successful election once it happens at the node with maximal ID. However, this is impossible, since the same event also occurs for every other node before the election process has finished.  

To be able to detect termination, we require an event that occurs \textit{uniquely} for the leader, and never for other nodes.
If we managed to run \cref{alg:nonterminating_le} over the \CCW channel after a full execution of the same algorithm over
the \CW channel, then, by symmetry, all nodes would eventually
receive \maxID many \CCW pulses.
In this scenario, the event 
$\totrecvcw[v] = \ID{v} = \totrecvccw[v]$ would, in fact, be unique to the node with the largest ID. Indeed, an initial full execution of \cref{alg:nonterminating_le} over the \CW channel guarantees that all other nodes have $\totrecvcw[v] > \ID{v}$ prior to exchanging \CCW pulses, so the above event occurs uniquely at the leader and could be used as the trigger for termination.

The main remaining difficulty is that we \emph{cannot} start the \CCW algorithm 
after the \CW one is done since neither the leader nor any other node can infer that the \CW algorithm has stabilized purely from 
the number of \CW pulses received.
We overcome this obstacle by running both algorithms in parallel, ensuring that the \CCW one lags behind the \CW one. By subtly prioritizing the execution of the \CW algorithm over that of the \CCW one, we enforce that once
$ \totrecvccw[v] = \ID{v} $ occurs for some \nonleader node, we already have $\totrecvcw[v] > \ID{v}$. Then, the only node $v$ satisfying $\totrecvcw[v] = \ID{v} = \totrecvccw[v]$ is the elected leader.
It is the uniqueness of all IDs, crucially including~\maxID, that enables this approach.

\begin{algorithm}[ht]   
\caption{Quiescently Terminating Leader Election for Node $v$}\label{alg:terminating_le}
    \begin{algorithmic}[1]
        \State{\SendCW{}}
        \Repeat
            \If{\RecvCW{} returns 1} 
            \label{alg:LE:CW-begin}
            \Comment{Run \Cref{alg:nonterminating_le} over the \CW channel}
                \If{$\totrecvcw = \ID{v}$} \label{alg:LE:if-totCW=ID}
                \State    {$\state \gets \leader$}
                \Else
                    \State{$\state \gets \nonleader$}
                    \State{\SendCW{}}
                \EndIf
            \EndIf  \label{alg:LE:CW-end}

            \Statex
            
            \If{$\totrecvcw \geq \ID{v}$}
            \label{alg:LE:CCW-begin}
            \Comment{Run \Cref{alg:nonterminating_le} over the \CCW channel, once $\totrecvcw \geq \ID{v}$}
                \IIf{$\totsentccw = 0$} {\SendCCW{}}
                \EndIIf
                \If{\RecvCCW{} returns 1}
                    \If{$\totrecvccw = \ID{v}$}
                        \State   {\textbf{pass}}
                    \Else
                        \State \SendCCW{}
                    \EndIf
                \EndIf
            \EndIf
            \label{alg:LE:CCW-end}
            
            \Statex
            
            \If{$\totrecvcw = \ID{v} = \totrecvccw$}
            \Comment{Initiate a termination pulse}
            \label{alg:LE:term-begin}
                \State{\SendCCW{}}  \label{line:terminating_pulse_}
                
                \While {\RecvCCW{} returns 0} \State {\textbf{pass}}\label{line:leader_waits} \Comment{Wait for return of termination pulse}
                \EndWhile
            
            \EndIf  
            \label{alg:LE:term-end}
        \Until{$\totrecvccw > \totrecvcw$} \label{line:terminating_condition_}
        \State \textbf{output} $\state$
    \end{algorithmic}
\end{algorithm}

The complete algorithm with quiescent termination is given in  \Cref{alg:terminating_le}.
It is composed of two instances of \Cref{alg:nonterminating_le}: one over the \CW channel (\crefrange{alg:LE:CW-begin}{alg:LE:CW-end}) and one over the \CCW channel (\crefrange{alg:LE:CCW-begin}{alg:LE:CCW-end}).
The \CW instance starts, as before, upon initialization, while the \CCW instance starts at node~$v$ once it reaches the
$\totrecvcw[v] = \ID{v}$ event in the \CW instance. This guarantees that the \CCW instance lags behind the \CW one. 
Finally, \crefrange{alg:LE:term-begin}{alg:LE:term-end} identify the event $\totrecvcw[v] = \ID{v} = \totrecvccw[v]$ that uniquely occurs at the leader. At that point of time, all the nodes have $\totrecvcw$ and $\totrecvccw$ set to~$\maxID$, and the leader is the last node for which this event occurs. This triggers the termination process: the leader sends a single \CCW pulse. Any node receiving this extra pulse sees, for the first time, $\totrecvccw > \totrecvcw$, forwards the pulse and terminates 
(\cref{line:terminating_condition_}). The extra pulse is forwarded until it returns to the leader, causing it to terminate without forwarding the pulse.
The analysis of \cref{alg:terminating_le} is deferred to \cref{sec:deferred-analysis-terminating-le} due to the space constraints.

\section{Leader Election in Non-Oriented Rings}
\label{sec:LE-non-oriented}

A natural follow-up question to the above leader election algorithm in oriented rings is whether the same holds for \emph{non-oriented} rings. 
In this setting, nodes do not possess a predefined \CW channel and \CCW channel anymore. Instead, they have two ports, \Port0 and \Port1, connecting them to their two neighbors in an arbitrary order.

The straightforward approach would be to first orient the ring with a quiescently terminating algorithm and then use our leader election algorithm from Section~\ref{sec:LE-oriented}. 
Since a ring orientation is easily computable once a leader has been elected with quiescent termination, orientation and leader election are essentially equivalent tasks from the perspective of quiescently terminating algorithms.

In this section, we instead present a quiescently \emph{stabilizing} algorithm for non-oriented rings, which both elects a leader and orients the ring. Recall that the difference from quiescent termination is that the nodes do not need to terminate explicitly; it suffices for all pulse activity to cease with the correctly computed output still present. Recall our main theorem for this part.

\orientRing*

We emphasize again that this algorithm does not terminate in the usual sense. Instead, its success is defined as reaching quiescence while guaranteeing that at that time only a single node has set its internal state to \leader, while all other nodes have set their state to  \nonleader.
Additionally, we require that nodes achieve a consistent ring orientation as follows: each node has to label exactly one of its two ports as the port leading to the \CW neighbor such that starting at some node and repeatedly moving to node connected to \CW port lets us pass through all edges in the ring. 

For ease of exposure, we first present an algorithm of slightly worse complexity but whose analysis is almost fully captured by results from earlier sections. We then improve its complexity by proving an additional property about \cref{alg:nonterminating_le}, our main building block for our algorithm for non-oriented rings (\cref{alg:ring_orientation}). Namely, we show that executing \cref{alg:nonterminating_le} on a ring with non-unique IDs essentially achieves the same guarantees as when used on a ring with unique IDs (\cref{lem:nt-le-duplicate-ids}). This observation allows us to halve the complexity of \cref{alg:ring_orientation}, as we shall see, and also has implications for solving the same tasks on anonymous rings with access to randomness.
As this extension to anonymous rings is a minor effort and follows from standard techniques, we defer it to \cref{sec:anonymousRings}.

\subparagraph*{Algorithm overview.}
At a high level, \cref{alg:ring_orientation} consists of two parallels executions of 
\cref{alg:nonterminating_le}, each one using each channel in the ring in a single direction. For an intuition of how that is possible, consider a setting where the network has a single pulse in transit. Suppose that all nodes execute the same algorithm that sends a pulse on \Port1 whenever one is received on \Port0, and vice versa. Forwarding the pulse in this manner has it travel the entire ring, since every time a pulse is received by a given node $v$ from one of its neighbors, it is sent to its other neighbor. If adding another pulse going in the other direction to the network, the two pulses independently travel around the ring in opposite directions. The nodes can thus effectively run two algorithms in parallel without them interfering with one another, as long as one only works with clockwise pulses and the other with counterclockwise pulses. As \cref{alg:nonterminating_le} only uses pulses going in one direction, it satisfies this requirement. The major caveat is that the nodes cannot be certain which of the two algorithms they execute is working with clockwise pulses and which is working with counterclockwise pulses.

Our algorithm has essentially two parts: one in which the node reacts to pulses and possibly forwards them (Lines~\ref{line:nt-le-pulse-start} to \ref{line:nt-le-pulse-end}), and one in which it computes its output depending on the number of pulses it received from each port (Lines~\ref{line:nt-le-output-computing-start} to \ref{line:nt-le-output-computing-end}). From the point of view of analyzing how many pulses are eventually sent in the network, only the first part is relevant. It is the part that simulates two executions of \cref{alg:nonterminating_le}.
For the nodes to settle on a consistent orientation of the ring in the second part, we break symmetry between the two options by having strictly more pulses sent in one orientation of the ring than in the other.

We distinguish the two parallel executions of \cref{alg:nonterminating_le} by having each node $v$ pick two distinct virtual IDs for itself, $\IDsup{v}{0}$ and $\IDsup{v}{1}$ (\cref{line:nt-le-new-ids}). $\ID{v}^{(1)}$ affects how $v$ behaves regarding pulses received from its \Port0, and symmetrically for $\ID{v}^{(0)}$ and \Port1. While nodes do not know which of their two virtual IDs is used in the clockwise or counterclockwise execution of \cref{alg:nonterminating_le}, they use a distinct ID in both executions. 
Eventually, in each direction, the number of pulses received by each node stabilizes to the largest ID in that direction. The choice of IDs ensures that the two parallel executions have distinct largest IDs, so eventually all nodes see strictly more pulses being sent in one direction than the other. This allows nodes to agree on a common orientation of the ring, and elect as leader the node who was the source of the largest ID.

The nodes use the following methods for sending and receiving pulses. Let $i\in\set{0,1}$.  
\begin{enumerate}
    \item \SendPort{i}{}: sends a pulse through \Port{i},
    \item \RecvPort{i}{}: check whether a pulse is waiting in the incoming queue of \Port{i}. If not, return~0. Otherwise, consume a single pulse from the queue and return~1.
\end{enumerate}

\begin{algorithm}[ht]\caption{Quiescently Stabilizing Leader Election on Non-Oriented Rings for Node $v$}
\label{alg:ring_orientation}
\begin{algorithmic}[1]
        \For{$i \in \set{0,1}$}
            \State{$\ID{v}^{(i)} \gets 2\cdot \ID{v} - 1 + i$}\label{line:nt-le-new-ids}
            \State{\SendPort{i}{}} \label{line:nt-le-first-pulse}
        \EndFor
        \While {\textsf{true}}
            \For{$i \in \set{0,1}$}\label{line:nt-le-pulse-start}
                \If {\RecvPort{1-i}{} returns 1 \textbf{and} $\totrecvport{1-i} \neq \ID{v}^{(i)}$} \label{line:nt-le-forward-test}
                        \State{\SendPort{i}{}} \Comment{Pulses received at one port are sent forward at the opposite one}\label{line:nt-le-forward-send}
                \EndIf
            \EndFor\label{line:nt-le-pulse-end}
            \If {$\max(\totrecvport{0},\totrecvport{1}) \geq \ID{v}^{(1)}$}\label{line:nt-le-output-computing-start}
                \If { $\totrecvport{0} = \ID{v}^{(1)}$ \textbf{and} $\totrecvport{1} < \ID{v}^{(1)}$}
                    \State{state $\gets$ \leader}
                \Else
                    \State{state $\gets$ \nonleader}
                \EndIf
                \If{$\totrecvport{0} > \totrecvport{1}$}\label{line:nt-le-orient}
                    \State{name $\Port0 \coloneqq \CCW$ and $\Port1 \coloneqq \CW$}
                    \Comment{I.e., $\Port0$ connects the \CCW neighbor}
                \Else
                    \State name $\Port0 \coloneqq \CW$ and $\Port1 \coloneqq \CCW$
                \EndIf
            \EndIf\label{line:nt-le-output-computing-end}
        \EndWhile
\end{algorithmic}
\end{algorithm}

\begin{proposition}
    \label{prop:suboptimal-nt-le}
    \Cref{alg:ring_orientation} elects a leader and consistently orients the ring in $n(4\cdot\maxID - 1)$ pulses. It achieves quiescence but does not terminate. 
\end{proposition}
\begin{proof}
    Consider $\vmax$, the node of largest ID.
    We define as clockwise the direction of a pulse sent from $\vmax$'s \Port1 and forwarded by all other nodes (i.e., sent from \Port{i} after arriving at \Port{1-i}). We call the ports sending such a pulse clockwise, and those receiving it counterclockwise.
    We show the our algorithm elects $\vmax$ as a leader and all nodes declare the correct port as clockwise.

    We show that the network eventually achieves quiescence and all nodes receive the same number of clockwise and counterclockwise pulses. Let us argue this for clockwise pulses; the property for counterclockwise pulses will follow by symmetry. For every node $v$ whose \Port1 connects it to its clockwise neighbor, let us rename \SendPort{1}{} to \SendCW{} and \RecvPort{0}{} to \RecvCW{} in its code. Let us also define $\IDcw{v} = \IDsup{v}{1}$ for such nodes. For the other nodes, which are connected to their clockwise neighbors through \Port0, rename \SendPort{0}{} to \SendCW{} and \RecvPort{1}{} to \RecvCW{} in their code, and let $\IDcw{v} = \IDsup{v}{0}$. We emphasize that this renaming is done purely for our analysis and is not an operation performed by the nodes, which are oblivious to what the clockwise direction is. Consider an execution of \cref{alg:ring_orientation}. Whenever the scheduler delivers a clockwise pulse, this pulse is read by a \RecvCW{} method, and if forwarded (which depends on $\IDcw{v}$), it is re-sent by a \SendCW{} method, according to our renaming. Other methods are never activated by clockwise pulses and never emit a clockwise pulse. As such, \cref{alg:ring_orientation} executes the exact same code on clockwise pulses as \cref{alg:nonterminating_le} and thus has the same guarantees as \cref{alg:nonterminating_le} regarding clockwise pulses. By \cref{cor:sent_recv_maxid}, this means that we achieve quiescence for clockwise pulses, with all nodes eventually having sent and received exactly the same number $\max_v \IDcw{v}$ of clockwise pulses. The same holds symmetrically for $\max_v \IDccw{v}$ pulses counterclockwise pulses per node, where $\IDccw{v}$ is defined similarly to $\IDcw{v}$.

    Since $\vmax$ picks as identifiers $2\cdot \maxID$ and $2\cdot \maxID-1$ for the two directions, we have $\max_v \IDcw{v} = 2\cdot \maxID$ and $\max_v \IDccw{v} = 2\cdot \maxID-1$. Hence, all nodes have sent and received $2\cdot \maxID$ clockwise pulses and $2\cdot \maxID-1$ counterclockwise pulses. This yields a bound of $n(4\cdot \maxID - 1)$ pulses, ensuring a consistent orientation according to the test on \cref{line:nt-le-orient}.
\end{proof}

\subparagraph*{Improving the message complexity.}
We now show how to improve the complexity of \cref{alg:ring_orientation} to $n  (2\cdot\maxID + 1)$. 
Since nodes send pulses according to their IDs,  \cref{line:nt-le-new-ids} effectively doubles the number of pulses sent by the algorithm. To avoid this doubling, one can generate the two IDs in a different manner; for instance,
$\ID{v}^{(1)} \gets \ID{v} + 1$, and $\ID{v}^{(0)} \gets \ID{v}$.
However, this leads to assigning the same ID to multiple nodes. We argue that \cref{alg:ring_orientation} works correctly even when IDs are not unique as long as the largest clockwise and counterclockwise IDs are different. 
To that goal, we need to re-analyze \cref{alg:nonterminating_le} in such case, which we do in the next two technical lemmas.

\begin{lemma}
\label{lem:nt-le-duplicate-ids}
     \Cref{cor:sent_recv_maxid} still holds if nodes run \cref{alg:nonterminating_le} with non-unique IDs. This includes the case in which multiple nodes $v$ have $\ID{v} = \maxID$.
\end{lemma}

Note that \cref{lem:nt-le-duplicate-ids} is about \cref{alg:nonterminating_le}, the implicit main subroutine of \cref{alg:ring_orientation}, not \cref{alg:ring_orientation} itself. Most elements of the proof of \cref{cor:sent_recv_maxid} make no reference to the node of largest $\ID{}$. Most importantly, the invariants of \cref{lem:le_invariants} (that $\totsentcw[v] = \totrecvcw[v] + 1$ while $\totrecvcw[v] < \ID{v}$, and $\totsentcw[v] = \totrecvcw[v]$ once $\totrecvcw[v] \geq \ID{v}$) are consequences of how each node reacts to the pulses it receives, and changing the distribution of $\ID{}$s does not change that. \Cref{lem:leader_last}, however, makes an explicit reference to a node of largest $\ID{}$, which requires us to change the argument somewhat. \Cref{lem:multileader-last}, which we shall now prove, generalizes \cref{lem:leader_last} to the setting with non-unique $\ID{}$s.

\begin{lemma}
    \label{lem:multileader-last}
    Consider the set of nodes of largest $\ID{}$, $\Vmax = \set{v : \ID{v} = \maxID}$. In \cref{alg:nonterminating_le}, if
    $\totrecvcw[v] \geq \ID{v}$ for all $v \in \Vmax$ at some point, then $\totrecvcw[v] \geq \ID{v}$
    holds for every node $v$ at this point. That is, one of the nodes in $\Vmax$ is the last node to 
    satisfy $\totrecvcw[v] \geq \ID{v}$.
\end{lemma}
\begin{proof}
    Let $m = \card{\Vmax}$ be the number of nodes that hold~$\maxID$. 
    Denote these nodes $v_1,\ldots,v_m$, ordered clockwise from an arbitrary one of them.
    Let us identify $v_{m+1} = v_1$ for ease of notation. 
    For each $i \in [m]$, let $k_i \in \set{0,\ldots,n-1}$ be the number of consecutive nodes with $\ID{v} < \maxID$ preceding $v_i$ in the ring. 
    For each $j \in [k_i]$, let $x_{i,j}$ be the node $j$ counterclockwise hops from $v_i$ in the ring. See \cref{fig:naming} for an illustration.
    
    We show that if $\totrecvcw[v_i] \geq \ID{v_i}$ holds at $v_i$, then it also holds at all $x_{i,j}$, $j\in [k_i]$. Since every node $v\not \in \Vmax$ has a node of largest $\ID{}$ in its clockwise direction later in the ring, this implies that when $\totrecvcw[v] \geq \ID{v}$ holds at all $v_i \in \Vmax$, it also holds at every node $v \in V$.
    
    \begin{figure}[ht]
        \centering
        \includegraphics[width=0.6\linewidth,page=2,trim={0 0 0 0}]{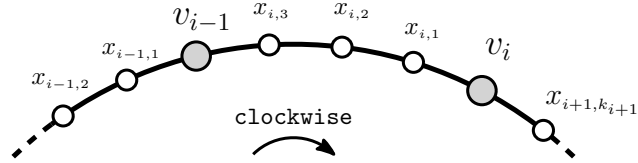}
        \caption{The naming of nodes between nodes of largest $\ID{}$ in the proof of \cref{lem:multileader-last}.}
        \label{fig:naming}
    \end{figure}

    Consider $v_i \in \Vmax$ s.t.\ $\totrecvcw[v_i] \geq \ID{v_i} = \maxID$. Suppose $k_i > 0$, as otherwise the result is trivial. Let $x_{i,0} = v_i$. For each $j \in [0,k_i)$, we show that $\totrecvcw[x_{i,j}] \geq \maxID$ implies $\totrecvcw[x_{i,j+1}] \geq \maxID$. As the base case $\totrecvcw[x_{i,0}] = \totrecvcw[v_i] \geq \maxID$ holds, we get our result by induction.
    
    Let $j \in [0,k_i)$ and suppose $\totrecvcw[x_{i,j}] \geq \maxID$. Since $\totrecvcw[x_{i,j}] \leq \totsentcw[x_{i,j+1}]$, we have that $\totsentcw[x_{i,j+1}] \geq \maxID$. Since $\maxID > \ID{x_{i,j+1}}$, by \cref{lem:le_invariants}, it needs to hold that $\totrecvcw[x_{i,j+1}] = \totsentcw[x_{i,j+1}]$. Therefore, $\totrecvcw[x_{i,j+1}] \geq \maxID$.
\end{proof}

Equipped with \cref{lem:multileader-last},  the proof of \cref{lem:nt-le-duplicate-ids} follows quite naturally.

\begin{proof}[Proof of \cref{lem:nt-le-duplicate-ids}]
    Let us review the proof of \cref{cor:sent_recv_maxid}, and see which elements of it could be affected by multiple nodes having the same $\ID{}$. Let $\Vmax = \set{v: \ID{v} = \maxID}$ be the set of nodes of largest $\ID{}$.
    
    As already stated, the invariants of \cref{lem:le_invariants} still hold, and arguments relying on \cref{lem:leader_last} must be amended to rely on \cref{lem:multileader-last} instead.
    \Cref{lem:quiescence_to_sat,lem:sat_to_quiescence,cor:quiescence_convergence}, about how the network is in quiescence if and only if each node~$v$ has received (and, by \cref{lem:le_invariants}, also sent) at least $\ID{v}$~pulses, are immediate consequences of \cref{lem:le_invariants}, and hence still hold.

    We now get to \cref{lem:le_equivalences,lem:le_convergence}, which are two lemmas from which \cref{cor:sent_recv_maxid} follows most directly.
    \Cref{lem:le_equivalences} shows equivalences between three properties: quiescence, that $\totrecvcw[v] \geq \ID{v}$ at each $v$, and that $\totrecvcw[v] = \totsentcw[v] = \maxID$ at each $v$. Again, the arguments still hold if $\ID{}$s are not unique: the connection between quiescence and all nodes satisfying $\totrecvcw[v] \geq \ID{v}$ was shown in \cref{cor:quiescence_convergence}; while some node $v \in \Vmax$ has received less than $\maxID$ pulses, the network is not in quiescence; for all nodes in $\Vmax$ to have received $\maxID$ pulses, other nodes in the network need to have sent this many pulses, which implies they have also received $\maxID$ pulses.

    \Cref{lem:le_convergence} shows that the network cannot permanently remain in a state in which some nodes have $\totrecvcw[v] < \ID{v}$. The argument again does not rely on the uniqueness of the $\ID{}$s, and only uses that if some non-empty set $B$ of nodes satisfies $\totrecvcw[v] < \ID{v}$ for each $v \in B$, then since for those nodes $\totsentcw[v] = \totrecvcw[v] + 1$ (\cref{lem:le_invariants}) the network has pulses still in transit. 
    These pulses must eventually reach nodes in~$B$, increase the number of received pulses $\totrecvcw$ of nodes in it, and eventually to the point that for a node in $B$, $\totrecvcw[v] \geq \ID{v}$.

    Put together, the whole argument still holds with non-unique $\ID{}$s.
\end{proof}

The proof of our second main theorem is a corollary of the above.
\begin{proof}[Proof of \cref{thm:orient-ring}]
    Let us modify \cref{alg:ring_orientation} as follows:
    \begin{itemize}
        \item In \cref{line:nt-le-new-ids}, we set as $\ID{}$s $\ID{v}^{(1)} \gets \ID{v} + 1$, and $\ID{v}^{(0)} \gets \ID{v}$.
    \end{itemize}
    Similar to the argument in the proof of \cref{prop:suboptimal-nt-le}, this amounts to running two parallel instances of \cref{alg:nonterminating_le} over each channel.
    As in that proof, consider the maximal clockwise and counterclockwise IDs, $\max_v \IDcw{v} = \maxID+1$ and $\max_v \IDccw{v} = \maxID$. From \cref{lem:nt-le-duplicate-ids}, we know that the number of clockwise pulses sent and received by each node eventually stabilizes at $\maxID+1$, and similarly stabilizes at $\maxID$ for the counterclockwise direction. 
    This results in a single leader being elected and a consistent orientation as before. The upper bound of $n(2\maxID + 1)$ follows from $n (\maxID + 1)$ pulses being exchanged in one direction, and $n \cdot \maxID$ in the other.
\end{proof}

\section{Lower Bound On Message Complexity in Content-Oblivious Rings}
\label{sec:lower-bounds}
In this section, we show that the  dependency  of the message complexity of our algorithms on $\maxID$ is natural and inevitable by providing a lower bound showing that the number of pulses sent increases indefinitely with the number of available IDs. 
\begin{theorem}\label{thm:lowerbound}
Let $k$ and $n$ be arbitrary positive integers, $k\ge n$. 
If $k$ distinct IDs are assignable to the $n$ nodes of the ring, 
at least $n\lfloor\log (k/n)\rfloor$ pulses are sent by any leader election algorithm for some assignment of IDs.
In particular, an unbounded number of pulses is sent for an infinite supply of IDs even on rings with just a single node. 
\end{theorem}

The proof of this theorem makes use of the following definition.
\begin{definition}[Solitude pattern]
Consider a ring with a single node ($n=1$), and fix a specific algorithm. 
Assume a scheduler that delivers pulses  one by one, keeping the order in which they were sent (breaking ties by prioritizing \CW pulses). 
Define the \emph{solitude pattern} as the sequence of incoming pulses observed by the node, encoded as 
a binary string where $0$ and $1$ 
encode \CW and \CCW pulses, respectively.
We denote the solitude pattern of a node with $\ID{}=i$ by $p_i$.
\end{definition}
Besides this crucial definition, we make use of the following lemma telling us that each ID has its own, unique solitude pattern.
Essentially, the proof relies on  matching all possible pairs of IDs against each other in a ring of two nodes. If any two IDs had the same solitude pattern, they would send and receive pulses in this ring ($n=2$) exactly as they would in solitude ($n=1$). Thus, in one of these execution they give an invalid output.
\begin{lemma}
\label{lem:p-is-unique}
    For any uniform content-oblivious leader election algorithm, each solitude pattern is unique. 
    In other words, for any pair of distinct IDs $i\neq j$, we have $p_i \ne p_j$.
\end{lemma}

\begin{proof}
Fix a content-oblivious leader election algorithm on uniform rings.
Seeking contradiction, assume that two nodes with distinct IDs, $i\neq j$, have the same solitude pattern, $p_i= p_j$. 
Note that each of these nodes outputs \leader when run in isolation ($n=1$). 

Consider a ring with $n=2$ nodes, which are 
assigned the IDs $i$ and~$j$, respectively. 
Assume a scheduler behaving the same way as in the definition of a solitude pattern for each node individually. That is, pulses arrive one by one in the order they were sent out. 
Moreover, we assume that the same delay is applied to all pulses.  
Since the two nodes' solitude patterns are identical, and the scheduler maintains order, each node will receive (and thus generate) exactly its solitude pattern.
Indeed, since $p_i=p_j$, both nodes send their first pulse in the same direction, and hence both receive it from the same direction as they would when alone and then send their next pulse accordingly; thus both receive and send exactly the pattern~$p_i=p_j$.
This means that both nodes output \leader as they must in their solitude situation. This contradicts the guarantees of the leader election task.
\end{proof}

Having established that each ID has its own unique solitude pattern, a lower bound arises from properties of binary strings, namely the length required to avoid repeating patterns.
We begin with the following simple property following from the pigeon-hole principle.
\begin{lemma}\label{lem:string-pigeon-hole}
For any $\stringlength,n\ge 1$, any set of $n2^\stringlength$ distinct binary strings contains $n$ strings sharing a common prefix of length at least~$\stringlength$.
\end{lemma}
\begin{proof}
Let $\stringlength$ and $n$ be positive integers. 
There are only $2^\stringlength-1$ distinct binary strings shorter than~$\stringlength$. 
Therefore, in a set of $n2^\stringlength$ distinct binary strings at least $n2^\stringlength-2^\stringlength-1=(n-1)2^\stringlength+1$ of them have length at least $\stringlength$, implying that they contain a prefix string of length $\stringlength$.
There are only $2^\stringlength$ distinct binary prefixes of length $\stringlength$, thus the pigeon-hole principle implies that at least one prefix of length $\stringlength$ is 
shared by at least $\lceil ((n-1)2^\stringlength+1)/2^\stringlength\rceil=\lceil n-1+2^{-\stringlength}\rceil=n$ of the strings in the set. 
\end{proof}

After deducing another corollary, we are ready to prove the lower bound of Theorem~\ref{thm:lowerbound}.

\begin{corollary}
\label{cor:common-prefix}
For any integers $k\ge n\ge 1$, any set of $k$ distinct binary strings contains $n$ strings sharing a common prefix of length at least~$\lfloor\log (k/n)\rfloor$.
\end{corollary}

\begin{proof}
Apply Lemma~\ref{lem:string-pigeon-hole} with the given $n$ and $\stringlength=\lfloor\log (k/n)\rfloor$. Note that the given set is large enough since $n2^\stringlength\le n2^{\log(k/n)}=n(k/n)=k$.
\end{proof}

\begin{proof}[Proof of Theorem~\ref{thm:lowerbound}]
Assume that we have a uniform leader election algorithm, a ring with $n$ nodes, and at least $k$ assignable IDs for any positive integers $k\ge n$. 
Due to Lemma~\ref{lem:p-is-unique} we know that each ID has its own unique solitude pattern, which is just a binary string. 
By \cref{cor:common-prefix}, there are at least $n$ IDs whose solitude patterns share a common prefix of length~$\stringlength=\lfloor\log (k/n)\rfloor$.
Assume a scheduler that behaves as described in Lemma~\ref{lem:p-is-unique}, for all $n$~nodes. 
It follows that all nodes send and receive pulses in exactly the same way as in their respective solitude situations for the first $\lfloor\log (k/n)\rfloor$ time steps, sending one pulse in each time step. Consequently, at least $n\cdot\stringlength=n\lfloor\log (k/n)\rfloor$ pulses are sent in total, proving the first part of the theorem. The last statement in the theorem follows for an infinite set of assignable IDs because $n\lfloor\log (k/n)\rfloor$ grows indefinitely with increasing $k$, even for $n=1$.
\end{proof}
Theorem~\ref{thm:LB} is then an immediate corollary of the above since the number of distinct IDs is bounded by~\maxID. This lower bound complements the upper bound of Theorem~\ref{thm:LE-oriented} and proves that the $\maxID$ term is not an artifact of our analysis or algorithm design but, rather, an inherent property of the problem in this setting.

\section{Conclusion and Open Questions}\label{sec:conclusion}
As our main result, we have presented a quiescently terminating algorithm for leader election in oriented rings with unique IDs that communicates $n(2\cdot\maxID+1)$ pulses. 
This implies that any content-oblivious computation can be performed on rings without assuming a pre-existing leader.
We have also provided a lower bound showing that the message complexity depending on \maxID is not a fluke but inherent to the problem. 

An immediate candidate for future work is to extend our results
from rings to general networks, i.e., to design a content-oblivious leader election algorithm in arbitrary 2-edge connected networks or, alternatively, prove this task impossible. 
Considering non-oriented rings may be useful towards that goal since there is no sense of direction in general networks.
Our content-oblivious leader election for non-oriented rings does not terminate, and we conjecture that this is inherent to the model. 
It remains as an open task for future work to prove this or find a terminating algorithm. 

\bibliographystyle{plain}
\bibliography{network}

\clearpage
\appendix
\section*{Appendix}

\section{Deferred Proofs Regarding Achieving Quiescent Termination}
\label{sec:deferred-analysis-terminating-le}

In this section, we provide an analysis of \cref{alg:terminating_le} showing that it satisfies the conditions of \cref{thm:LE-oriented}, proving the theorem.

\LEoriented*

We begin with the following lemma, which asserts certain conditions on the number of pulses received over each channel. The lemma suggests that the \CCW instance cannot advance beyond the \CW instance. We use \textit{\CW-quiescence} to denote 
the less general state of quiescence with respect to the \CW channel only, that is, the network is in \CW-quiescence if
it no longer contains \CW pulses. The network then reaches quiescence if both \CW-quiescence and \CCW-quiescence are
achieved, where the latter is defined analogously.

\begin{lemma}\label{lem:progress}
    As long as \CW-quiescence has not been reached,
    $\totrecvccw[v] \le \totrecvcw[v]$ holds for every node~$v$, with equality
    only when $\totrecvcw[v] = 0$.
\end{lemma}
\begin{proof}
    First, note that if a node $v$ has $\totrecvcw[v] < \ID{v}$, then it has not yet participated in the
    \CCW algorithm, and $\totrecvccw[v] = 0$. Since $\totrecvcw[v] \geq 0$, the statement holds for $v$.
    Further, this is always the case at the node~$\vmax$ for which $\ID{\vmax} = \maxID$. 
    Indeed, since \Cref{alg:terminating_le} consists of two independent instances of \Cref{alg:nonterminating_le}, one over each channel, 
    if we restrict the discussion to the \CW channel, the execution of \Cref{alg:terminating_le}
    is identical to an execution of \Cref{alg:nonterminating_le}, and we may invoke the properties proven in \Cref{sec:warmup}. 
    Then, by \Cref{lem:le_equivalences}, while 
    \CW-quiescence has not occurred, we have $\totrecvcw[\vmax] < \maxID = \ID{\vmax}$, and the statement holds for~$\vmax$.

    We next consider the other nodes and show the statement holds also when $\totrecvcw[v] \ge \ID{v}$.
    Fix an arbitrary node $v \neq \vmax$ with $\totrecvcw[v] \geq \ID{v}$. 
    By \Cref{lem:le_invariants:prop:two},
    $\totrecvcw[v] = \totsentcw[v]$.
    Now, let $u_1, \ldots, u_k$ be \CW successors of $v$ such that
    $u_k$ is the closest node to $v$ with $\totrecvcw[u_k] < \ID{u_k}$.
    Since $v \neq \vmax$, such a node $u_k$ exists.

    Since $\totrecvcw$ is monotonically non-decreasing, once the
    condition $\totrecvcw \geq \ID{v}$ on \Cref{alg:LE:CCW-begin} is met, it is
    always met for all future iterations. As \CCW pulses are buffered, i.e., never lost, 
    and because the nodes $\{v, u_1, \ldots, u_{k-1}\}$ meet the condition
    on \Cref{alg:LE:CCW-begin}, they can be
    viewed as running \Cref{alg:nonterminating_le} over the \CCW channel. As such, all the properties proven in the previous section hold for the \CCW instance as well
    \footnote{The only difference is that nodes start the \CCW instance in different times, while for the \CW instance we assumed all nodes start at the same time. However, this is not an issue. Pulses on the \CCW channel remain in the queue of a node that has not started yet. That is, the situation is equivalent to the case where all the nodes start at a later time and pulses are being delayed accordingly.},
    and we may apply \Cref{lem:le_invariants} over the \CCW channel.

    Consequently,
    $\totsentccw[u_i] \leq \totrecvccw[u_i] + 1$ for $i < k$. Moreover,
    $\totrecvccw[u_i] \leq \totsentccw[u_{i+1}]$ for $i < k$, and $\totrecvccw[v] \leq \totsentccw[u_1]$.
    Combining these inequalities yields
    \[
        \totrecvccw[v] \leq \totsentccw[u_k] + k - 1\text{.}
    \]
    
    Since $\totsentccw[u_k] = 0$, we get $\totrecvccw[v] \leq k - 1$. 
    Let $u$ be the node with the
    largest ID out of $\{v, u_1, \ldots, u_{k-1}\}$. Since we have $k$ nodes, at least one
    must have an ID at least $k$ by the uniqueness of IDs, that is, $\ID{u} > k - 1$.
    
    If $u = v$, then we have  $\totrecvccw[v] < \ID{v} \leq \totrecvcw[v]$.
    On the other hand, if $u \neq v$, then $u = u_{k'}$ for some $k' < k$. Using 
    $\totrecvcw[u_{i+1}] \leq \totsentcw[u_{i}]$ for $i < k'$ and $\totsentcw[u_i] = \totrecvcw[u_i]$
    yields
    \[
    \totrecvccw[v] \leq k - 1 < \ID{u_{k'}} \leq \totrecvcw[u_{k'}] \leq  \totsentcw[u_{k'-1}] =   \totrecvcw[u_{k'-1}] 
    \leq \ldots \leq \totsentcw[v] = \totrecvcw[v].
    \]
    Therefore, in both cases we get $\totrecvccw[v] < \totrecvcw[v]$, which concludes the proof.
\end{proof}

A direct implication of \Cref{lem:progress} is that no node may terminate until the \CW phase has concluded since
no node may meet the termination condition on \Cref{line:terminating_condition_} and exit the loop.

\begin{corollary}\label{cor:cw_qu_non_termination}
    As long as \CW-quiescence has not been reached, nodes do not terminate.
\end{corollary}
\begin{proof}
    Initially, every node sends a \CW pulse.
    By \Cref{lem:progress},
    until there is \CW-quiescence, all nodes have $\totrecvccw \leq \totrecvcw$ and, hence, cannot terminate in \cref{line:terminating_condition_}.
\end{proof}

Now we prove a set of necessary conditions that hold if some $v$ meets
$\totrecvcw[v] = \totrecvccw[v] = \ID{v}$.
In particular,
$v$ is necessarily the node with the largest ID, every node is still awake with $\totrecvcw = \totrecvccw$,
and there must be quiescence in the network over both channels.

\begin{lemma}\label{lem:necessary_term_conditions}
    If the event $\totrecvcw[\vmax] = \totrecvccw[\vmax] = \ID{\vmax}$ occurs 
    for the first time at a node $\vmax$,
    then at that point of time:
    \begin{enumerate}
        \item the network has reached \CW-quiescence,
                                                    \label[lemma]{lem:necessary_term_conditions:prop:cw_quiescence}
        \item $\ID{\vmax} = \maxID$,                    \label[lemma]{lem:necessary_term_conditions:prop:l_is_maxid}
        \item there are no \CCW pulses in transit, i.e., $\sum_v{\totsentccw[v]} = \sum_v{\totrecvccw[v]}$,
                                                    \label[lemma]{lem:necessary_term_conditions:prop:ccw_quiescent}
        \item $\forall{v}: \totrecvcw[v] = \totrecvccw[v] = \maxID$, and
                                                    \label[lemma]{lem:necessary_term_conditions:prop:all_recv_maxid}
        \item no node has terminated.
                                                    \label[lemma]{lem:necessary_term_conditions:prop:non_termination}
    \end{enumerate}
\end{lemma}
\begin{proof}  
    If there is no quiescence over the \CW channel, then no node terminates by \Cref{cor:cw_qu_non_termination}.
    By \Cref{lem:progress}, 
    $\totrecvccw \leq \totrecvcw$ holds for all nodes with equality only when
    $\totrecvccw = \totrecvcw = 0$. Since $\ID{\vmax} > 0$, this proves the first implication.
    
    By \Cref{lem:le_equivalences} (applied to the \CW instance),  \CW-quiescence implies that all nodes $v$ satisfy
    $\totrecvcw[v] = \maxID$, and, in particular, due to the uniqueness of \maxID, only one
    node $\vmax$ with $\ID{\vmax} = \maxID$ satisfies $\totrecvcw[\vmax] = \ID{\vmax}$ at this point, which
    proves the second implication.

    Again, by \Cref{lem:le_equivalences}, \CW-quiescence is equivalent to all nodes $v$
    satisfying $\totrecvcw[v] \geq \ID{v}$. Thus, all nodes satisfy the condition on \Cref{alg:LE:CCW-begin}
    and partake in the \CCW instance.
    Now consider the point of time where $\totrecvcw[\vmax] = \totrecvccw[\vmax] = \ID{\vmax} = \maxID$
    holds for the first time. As $\vmax$ has not yet executed \Cref{line:terminating_pulse_}, an execution
    of \Cref{alg:terminating_le} is identical to an execution of \Cref{alg:nonterminating_le} over the
    \CCW channel. By \Cref{lem:leader_last} (applied to the \CCW channel at that point of time),
    the node~$\vmax$ is the last node
    to meet $\totrecvccw[\vmax] \geq \ID{\vmax}$, which is equivalent to $\forall{v}: \totrecvccw[v] = \totsentccw[v] = \maxID$
    by \Cref{lem:le_equivalences}, yielding the third implication.

    By a symmetric argument for the \CW instance, we get $\forall{v}: \totrecvcw[v] = \totsentcw[v] = \maxID$, 
    which proves the fourth implication.

    Finally, as argued previously, no node terminates until \CW-quiescence is achieved. Moreover, once there is
    \CW-quiescence, then $\totrecvcw[v] = \maxID$ remains to hold for all $v$ by \Cref{lem:le_equivalences}. As all nodes
    are partaking in the \CCW algorithm at this point, by \Cref{cor:recv_bound} applied to the \CCW instance,
    $\totrecvccw[v] \leq \maxID = \totrecvcw[v]$ holds
    until the event $\totrecvcw[\vmax] = \totrecvccw[\vmax] = \ID{\vmax}$ first occurs
    before $\vmax$ executes \Cref{line:terminating_pulse_},
    so no node has terminated, which concludes the proof.
\end{proof}

We now show that the condition $\totrecvcw[\vmax] = \totrecvccw[\vmax] = \ID{\vmax}$ 
in \Cref{lem:necessary_term_conditions} must be met at some point by a node
$\vmax$ possessing the largest ID, hence establishing the uniqueness of the event
we use to drive termination.

\begin{lemma}\label{lem:leader_starts_termination}
    Let $\vmax$ be the node with $\ID{\vmax} = \maxID$. Then there is an iteration
    where the event $\totrecvcw[\vmax] = \totrecvccw[\vmax] = \ID{\vmax}$ occurs.
\end{lemma}
\begin{proof}
    By \Cref{lem:necessary_term_conditions:prop:l_is_maxid}, the event cannot occur at any other node $v \neq \vmax$.
    Assume arguably that $\vmax$ never meets this condition, then all nodes are running exactly
    \cref{alg:nonterminating_le} over both channels. Thus, eventually, by \Cref{cor:sent_recv_maxid} applied
    to both instances, there is some point of time where every node has sent and received
    exactly $\maxID$ pulses over both channels. For node $\vmax$, we then have
    that $\totrecvcw[\vmax] = \totrecvccw[\vmax] = \maxID = \ID{\vmax}$ occurs, a contradiction.
\end{proof}

With the above, we can now conclude the proof of Theorem~\ref{thm:LE-oriented}.
\begin{proof}[Proof of Theorem~\ref{thm:LE-oriented}]
Assume an oriented ring of $n$ nodes with unique IDs, where each node executes \Cref{alg:terminating_le}.
Let $\vmax$ be the node with the largest ID. We claim that all nodes terminate at some point with exactly one node, $\vmax$,
outputting \leader while all other nodes output \nonleader.

 By \Cref{lem:leader_starts_termination}, 
$\vmax$ meets the condition $\totrecvcw[\vmax] = \ID{\vmax} = \totrecvccw[\vmax]$ at some
iteration, and, due to \Cref{lem:necessary_term_conditions:prop:l_is_maxid}, no other node does.
Consider the first such iteration.
By \Cref{lem:necessary_term_conditions:prop:non_termination}, no node has
terminated prior to that.
Additionally, $\totrecvcw[v] = \totrecvccw[v] = \maxID$ holds
for all $v$ due to \Cref{lem:necessary_term_conditions:prop:all_recv_maxid}. Thus, exactly one node, $\vmax$, still satisfies $\totrecvcw[\vmax] = \maxID = \ID{\vmax}$ and
has a state of \leader while every other node $v \neq \vmax$ has $\totrecvcw[v] = \maxID > \ID{v}$ and
a state of \nonleader. By \Cref{lem:necessary_term_conditions:prop:cw_quiescence},
the network is in \CW-quiescence, and no more \CW pulses
are sent or received, so those states are final.

It remains to show that all nodes terminate and that
when they do so, the network is in quiescence (no pulses still in transit, i.e., \emph{a quiescent termination}).

By \Cref{lem:necessary_term_conditions:prop:ccw_quiescent}, no \CCW pulses are in transit when the event
$\totrecvcw[\vmax] = \totrecvccw[\vmax] = \ID{\vmax}$ occurs firstly.
Once $\vmax$ executes \Cref{line:terminating_pulse_} during that iteration, there is exactly one \CCW
pulse in transit. Therefore, every node $v \neq \vmax$, having
$\totrecvcw[v] = \totrecvccw[v] = \maxID > \ID{v}$, receives this \CCW pulse, relays it, and then meets the termination
condition on \Cref{line:terminating_condition_} and outputs \nonleader correctly.
Finally, due to \Cref{line:leader_waits}, node $\vmax$ has been waiting on a \CCW pulse, which it then receives and consumes,
reaching \CCW-quiescence. The node $\vmax$ then meets the condition $\totrecvccw[\vmax] > \totrecvcw[\vmax]$
and terminates with the correct output of \leader. 
As quiescence has been reached over both channels, the algorithm terminates quiescently and correctly elects
a leader.

We finally describe the complexity of the algorithm. By termination,
every node $v$ in the algorithm has $\totrecvcw[v] = \maxID$ and
$\totrecvccw[v] = \maxID + 1$. Consequently, the overall complexity is
$n \cdot ({\maxID + 1}) + n \cdot \maxID = n (2\cdot \maxID + 1)$.
\end{proof}

\section{Anonymous Rings}
\label{sec:anonymousRings}
In this section, we consider the setting where a ring consists of $n$
identical nodes, each with access to an independent source of randomness.
We call such a ring \textit{anonymous}. As is standard in the literature about randomized algorithms, we aim to solve our computational task \emph{with high probability}, defined as bounding the probability of failure by an arbitrary inverse power of~$n$, the size of the network. That is, we present an algorithm parameterized by a value $c>0$, such that for any~$n$, the algorithm correctly elects a leader in $n$-sized rings with probability at least $1-O(n^{-c})$. Our algorithm also correctly orients a non-oriented ring with high probability.

As alluded to in the main text, electing a leader with quiescent termination is impossible in the context of this section, even if we relax the objective to only succeed with some arbitrary small constant probability. This follows from the following negative result by Itai and Rodeh~\cite{IR90}. 
\begin{corollary}[Corollary of {\cite[Thm.\ 4.2]{IR90}}]
    \label{cor:anonymous-impossibility}
    For any constant $\varepsilon>0$, there is no quiescently terminating algorithm for electing a leader in the anonymous ring that succeeds with at least constant probability $\varepsilon$.
\end{corollary}
\begin{proof}
    Itai and Rodeh show that no terminating algorithm can compute the number of nodes in an anonymous ring with any constant probability of success~\cite[Thm.\ 4.2]{IR90}. As any algorithm in our setting also works in the setting considered in that prior work (in which messages carry content instead of being pulses), the impossibility result carries to our setting. As any algorithm for electing a leader with quiescent termination would allow us to count the number of nodes in the ring afterwards~\cite{ccgs23}, leader election suffers from the same impossibility result.
\end{proof}

Consequently, we only aim for a quiescently stabilizing algorithm.

As in \cref{sec:LE-non-oriented}, 
observe that despite assuming the uniqueness of all IDs in \Cref{sec:warmup} to elect the node with the largest ID, the task remains
well-defined even if only the largest ID is unique. \Cref{lem:nt-le-duplicate-ids} showed that on a ring with possibly non-unique IDs, \cref{alg:nonterminating_le} elects as leader all nodes $v$ such that $\ID{v} = \maxID$. Hence, if a single node satisfies $\ID{v} = \maxID$, a single leader is elected. \Cref{alg:ring_orientation} is simply two parallel executions of \cref{alg:nonterminating_le}. As we already showed in the proof of \cref{thm:orient-ring}, those two parallel executions still yield the desired result as long as the maximal ID in the two executions is unique. Hence, providing an algorithm for sampling IDs with the guarantee that the maximal ID is unique with high probability is sufficient to obtain a variant of
\cref{thm:orient-ring} for anonymous rings.

In what follows, we outline a process by which nodes can utilize their access to
randomness to sample random IDs with the guarantee that, with high probability, the maximal ID is unique and of order $n^{O(1)}$. The algorithm terminates quiescently---in fact, it uses no communication---enabling composition: any algorithm can be performed afterwards with the sampled IDs. As a result, with high probability,
the anonymous setting reduces to the one considered in \Cref{lem:nt-le-duplicate-ids}. As explained previously, such an algorithm for sampling IDs immediately implies the following result.

\begin{restatable}{theorem}{AnonymousLEThm}
\label{thm:anon_le}
    There is a content-oblivious algorithm of complexity $n^{O(1)}$
    that elects a leader and orients an anonymous ring of $n$ nodes, each
    with access to its own source of randomness, with high probability.
    The algorithm reaches quiescence but does not terminate.
\end{restatable}

We now present the algorithm for sampling IDs guaranteeing a unique maximal ID with high probability; see \cref{alg:id_sampling}.
At a high level, the algorithm has each node first sample
the number of bits in its ID from a geometric distribution with parameter
$p = 2^{-1/\Theta(c)}$ before sampling said bits uniformly at random.
While each ID is of expected length $\Theta(c)$ and value $2^{\Theta(c)}$ at the end of this process, the maximal ID is of length $\Theta(c^2 \log n)$ and value $n^{\Theta(c^2)}$, with high probability. See \cite{bruss1990maximum}.
Intuitively, while sampling a large ID is an unlikely event, this unlikely event becomes more and more likely to occur somewhere in the network as the network grows. Similar ideas have previously appeared in the literature, e.g., in~\cite{AM94,FPSS22}.
While in other settings, as in~\cite{AM94}, similar ID sampling algorithms
may utilize message-passing to assign unique IDs, this
is not an option in a content-oblivious setting.
Fortunately, communication between nodes is not necessary in order
to just achieve a unique maximal ID with high probability.
Giving unique IDs to all nodes with quiescent termination is also impossible, even with just a small constant probability of success from the impossibility result of Itai and Rodeh~\cite{IR90} (see \cref{cor:anonymous-impossibility}). This is because unique IDs would allow us to run \cref{alg:terminating_le} and elect a leader. In fact, most nodes have non-unique IDs after running \cref{alg:id_sampling}.

\begin{algorithm}[ht]\caption{Message-Free Algorithm for Sampling an ID of Order $n^{O(c^2)}$ for Node $v$, $c > 0$}\label{alg:id_sampling}
		\begin{algorithmic}[1]
            \State{$p \gets 2^{-1/(c+2)}$}
            \State{Sample \textit{BitCount }$\sim Geo(1-p)$, i.e., according to the geometric distribution with
            parameter $1-p$}
            \State{Sample $\ID{v}$ uniformly at random from $\{0,1\}^{\textit{BitCount}}$}
		\end{algorithmic}
\end{algorithm}

\begin{lemma}\label{lem:id_sampling}
    For any constant $c > 0$, with high probability,
    running \Cref{alg:id_sampling} in an anonymous ring of $n$ nodes assigns each one an ID of size $n^{O(c^2)}$
    such that the maximal ID is attained uniquely by one node and is at least $n^{\Omega(c)}$.
\end{lemma}
\begin{proof}
    Let $S_1,\ldots, S_n$ be random variables associated to the
    the variable \textit{BitCount} for each node after termination
    is reached across all nodes.
    
    Observe that for all $i \in [n]$ and $x > 0$, we have $\Pr[S_i > x] = p^{x}$.
    
    Define the random variable $M := \max_i{S_i}$, that is, $M$ is the length of the longest
    sampled ID. We first aim to show that $M$ exceeds
    $L := (c+2)\log{\frac{n}{c\ln{n}}} = \Omega(c \log{n})$
     with high probability.

    Indeed, the complementary event $M \leq L$ has probability
        $\Pr[M \leq L] = \Pr[\forall i \in [n], S_i \leq L]
                               = (1-\Pr[S_i > L])^n 
                               = (1-p^{L})^n$,
    where the second equality follows due to the independence between
    the variables $S_i$.

    By definition of $p := 2^{-1/(c+2)}$ and the bound $1-x \leq e^{-x}$,
    we get
    \[
        \Pr[M \leq L] = (1-p^{L})^n 
        \leq \exp(-p^{L}n)
        \leq \exp(-c\ln{n}) = n^{-c}\text{.}
        \tag{1}\label{lem:max_len_lb}
    \]
    Thus, with high probability, $M > L = \Omega(c \log{n})$.
    To see that the largest ID is unique with high probability,
    note that the probability of a collision on $M$ bits is at most
    ${\binom{n}{2}} \cdot 2^{-M}$, which, with high probability due inequality~\eqref{lem:max_len_lb}, is bounded by 
    $\frac{n^2}{2} \cdot 2^{-L} =
                \frac{n^2}{2} \cdot (\frac{n}{c\ln{n}})^{-(c+2)} = O(n^{-c+\varepsilon})$
    for any small constant $\varepsilon > 0$.
    
    Finally, we show that $M = O(c^2 \log{n})$ with high probability.
    Letting $U := (c+1)(c+2) \log{n}$, the probability of $M$
    exceeding $U$ can be bound as follows:
    \[
        \Pr[M > U] = \Pr\left[\exists i \in [n], S_i > U\right]
        \leq p^{U} \cdot n = n^{-(c+1)} \cdot n = n^{-c} 
        \text{,}
    \]
    where the inequality follows from a union bound.
    Therefore, with probability at least $1 - n^{-c}$,
    the sampled IDs are bounded by $2^U = n^{O(c^2)}$.
\end{proof}

The proof of \Cref{thm:anon_le} follows immediately from \Cref{lem:id_sampling}.

Finally, we note that a slight modification of \cref{alg:ring_orientation} would also allow us to sample a unique ID for each node 
(\cref{prop:anonymous-unique-ids}), with high probability. As a result, the three settings of (1) the anonymous ring, (2) the ring with a leader, 
and (3) the ring with unique IDs and an orientation, are crucially separated by the possibility of quiescent termination. Leader election, orienting 
the ring, and assigning unique IDs can all be computed in setting (1). However, they can only be done without termination, while in settings (2) and (3), 
the same tasks and more can be performed with quiescent termination.

\begin{proposition}
    \label{prop:anonymous-unique-ids}
    Let $\ID{v}$ be the ID sampled by nodes in \cref{alg:id_sampling}, before running \cref{alg:ring_orientation}. Modify \cref{alg:ring_orientation} so that whenever a node receives a pulse, if $\min(\totrecvport{0},\totrecvport{1}) > \ID{v}$, node $v$ updates its ID to a new ID sampled uniformly at random between $1$ and $\min(\totrecvport{0},\totrecvport{1})-1$. Then, with high probability, all nodes have distinct IDs when reaching quiescence.
\end{proposition}
\begin{proof}
    First, remark that nodes always forward pulses once $\min(\totrecvport{0},\totrecvport{1})$ has exceeded $\ID{v}$. The way we update the ID keeps the ID below $\min(\totrecvport{0},\totrecvport{1})$, so this change does not change the behavior of nodes in the algorithm regarding how they send or receive pulses. It also does not affect the computation of the output.
    
    Each node last resamples its ID when $\min(\totrecvport{0},\totrecvport{1})$ reaches $\maxID$. As the largest ID is at least $n^{\Omega(c)}$ with high probability, all nodes end up sampling their IDs in a sampling space of size at least $n^{\Omega(c)}$. Therefore, the probability of a collision for the last sampled IDs is at most $\binom{n}{2}\cdot n^{-\Omega(c)} = n^{-\Omega(c)}$, i.e., there is no collision with high probability.  
\end{proof}
\end{document}